\documentclass[aps,pra,twocolumn,showpacs,superscriptaddress,groupedaddress]{revtex4-2}
\usepackage[utf8]{inputenc}
\usepackage{amsmath, amsthm, amssymb, amsfonts,enumitem}
\usepackage[makeroom]{cancel}
\usepackage{graphicx}
\usepackage{bm}
\usepackage{dsfont} 
\usepackage{mathtools}
\usepackage{booktabs}
\usepackage{thm-restate}
\usepackage[colorlinks=true, linkcolor=blue, urlcolor=blue,citecolor=blue]{hyperref}
\usepackage{orcidlink}
\usepackage{mathtools}
\DeclareMathOperator{\diag}{diag}

\DeclareMathOperator{\tr}{tr}

\DeclareMathOperator{\id}{id}

\newcommand{\GHZ}{{\rm{GHZ}}}
\newcommand{\bra}[1]{\mathinner{\langle #1|}}
\newcommand{\ket}[1]{\mathinner{|#1\rangle}}

\newcommand{\dyad}[1]{| #1\rangle \langle #1|}

\newcommand{\one}[0]{\mathds{1}}

\newcommand{\C}{\mathds{C}}
\newcommand{\WW}{\mathcal{W}}

\DeclareUnicodeCharacter{202C}{\^{i}}

\makeatletter
\newcommand{\vast}{\bBigg@{4}}
\newcommand{\Vast}{\bBigg@{5}}
\makeatother
\newtheorem{theorem}    {Theorem}

\newtheorem{proposition}[theorem]{Proposition}
\newtheorem{observation}[theorem]{Observation}

\usepackage{soul,xcolor}

\begin{document}

\title{
State -- witness contraction
}
\author{Albert Rico${}^{\orcidlink{0000-0001-8211-499X}}$}
\affiliation{
Physics department, Universitat Autònoma de Barcelona,
ES-08193 Bellaterra (Barcelona), Spain;}
\affiliation{
Faculty of Physics, Astronomy and Applied Computer Science, Institute of Theoretical Physics, Jagiellonian University,
30-348 Krak\'{o}w, 
Poland}
\date{\today}
\begin{abstract}

We construct both nonlinear and linear entanglement witnesses, by tensoring and partial tracing existing states and witnesses. We show that little shared quantum resources allow to employ decomposable witnesses to obtain larger ones detecting locally undistillable states, and find analytic witnesses for multipartite entangled states that are undistillable across all bipartitions. Crucially, we show how to efficiently optimize multicopy witnesses to desired states, which is a current challenge. As an example of both theoretical and experimental interest, existing trace polynomial witnesses are significantly 
improved while preserving their symmetries and implementability with randomized measurements. Besides detecting entanglement, we also find new witnesses detecting 
$k$-copy distillability. A recipe for the single-copy case is shown to be effective for generic and Werner states.
\end{abstract}

\maketitle

{\em Introduction.}\,\,
In the last decades, a large improvement has been done in manipulating mutlipartite and high-dimensional quantum systems~\cite{advancesHDent_Erhard2020,EfficientLargeScaleMBdyn_Artaco2024}. This paves a road towards realistic quantum software~\cite{MPEntSuperQubits_Lu2022}, quantum networks~\cite{LargScaleQNetworksKozlowski_2019,NetworkGMPE_Navascues2020} and quantum devices where entanglement constitutes a major nonclassical resource. As a result, detecting the presence of entanglement in multipartite quantum systems becomes crucial~\cite{MPEntSuperQubits_Lu2022,DetMPentManyBody_Frerot2022}, but it is at the same time a theoretical and experimental challenge due to the growth of the state space with the local dimension~\cite{GURVITS2003sepNPhard,DetMBodyContinuous_Kunkel2022}.

For this purpose, a vastly used technique is the evaluation of entanglement witnesses, namely observables whose negative expectation value detects the presence of entanglement~\cite{EntDetRev_Guhne2009,ReviewQEnt_Horo2009,bae2020mirrored}. However, important challenges remain: for some states, finding a witness detecting them requires unaccessible computational resources~\cite{doherty2004complete,Navascues_PPT_DPS_2009}; finding new entanglement witnesses for multipartite systems of interest is challenging~\cite{Huber2022DimFree,Dagmar_DetectHypergraph2017,Guhne_DetectGraph}; and adjusting a witness to a state cannot be done systematically in general~\cite{GURVITS2003sepNPhard}. A promising tool detecting larger sets of states~\cite{Liu_FundamLimDet2022} are nonlinear witnesses acting on multiple copies~\cite{HoroMulticopyWit2003,RemikUniversal2008}. While these overcome some limitations of linear witnesses~\cite{Liu_FundamLimDet2022,TP_Rico24}, problems like tailoring multicopy witnesses~\cite{HoroMulticopyWit2003,RemikUniversal2008,TP_Rico24} to states are even more challenging than in the linear case.

Here we provide a method to derive 
both linear and nonlinear entanglement witnesses using existing states and witnesses (Prop.~\ref{prop:StWitContr} and Fig.~\ref{fig:ContrWit}). 
We show that little shared entanglement suffices to lift existing witnesses to larger ones with further detection capabilities (Observation~\ref{obs:4PartWitDetBoundEnt}), and we construct simple witnesses detecting bound multipartite entanglement (Proposition~\ref{prop:MPPTwits}). 
By adjusting the constructed witnesses 
-- see Eq.~\eqref{eq:OptimWitGen}, we 
are able to optimize to target states a class witnesses acting on multiple copies~\cite{HoroMulticopyWit2003,RemikUniversal2008} (Fig.~\ref{fig:NonlinOpt} and Table~\ref{tab:DetectIso}). Besides detecting entanglement, the method is shown to be effective at detecting $k$-copy distillability -- see Eq.~\eqref{eq:Distillability}. 
\\

{\em State-witness contraction.}\,\,
The following result provides a systematic recipe to construct new entanglement witnesses from existing states and witnesses (depicted in Fig.~\ref{fig:ContrWit}). 
\begin{proposition}\label{prop:StWitContr}
    Let $\{R^{(j)}\}_{j=1}^n$ be states or witnesses in $k_j$-partite systems $S^j$. Let $\tau_{S^1_{k_1}...S^n_{k_n}}$ be a $n$-partite state or witness shared between the last parties of all systems $S^j$. Then, the operator
\begin{align}
  \WW_\tau := \tr_{S^1_{k_1}...S^n_{k_n}}\Big ( &R^{(1)}_{S^1_{1...k_1}}\otimes...\otimes R^{(n)}_{S^n_{1...k_n}}\,\cdot\label{eq:WitGenkn}\\
  &\one_{S^1_{1...k_1-1},...,S^n_{1...k_n-1}}\otimes\tau_{S^1_{k_1},...,S^n_{k_n}} \Big )\nonumber
 \end{align}
 has nonnegative expectation value on separable $\kappa$-partite states $\varrho_{S^1_{1...k_1-1},...,S^n_{1...k_n-1}}$, $\tr(\WW_\tau\varrho)\geq 0$, where $\kappa=k_1+...+k_n-n$.
\end{proposition} 
This can be seen from the duality between the block-positive cone and the separable cone~\cite{Karol-GeoQstates2006} as follows. 
\begin{proof}
Using the coordinate-free definition of the partial trace, for any state $\sigma^{(j)}_k$ and product state $\varrho^{(j)}=\varrho_1^{(j)}\otimes...\otimes\varrho_{k-1}^{(j)}$ we have $0\leq \tr(R^{(j)}\cdot\varrho^{(j)}\otimes\sigma^{(j)}_k)
=\tr (\theta^{(j)}\otimes\sigma^{(j)}_k )$, where we define  $\theta^{(j)}:=\tr_{1...k-1}(R^{(j)}\cdot\varrho_1^{(j)}\otimes...\otimes\varrho_{k-1}^{(j)}\otimes\one)$. Equivalently, $\theta^{(j)}$ is positive semidefinite for all $1\leq j\leq n$. Therefore, for any  $n$-partite state or witness $\tau$, we have the inequality $0\leq \tr(\theta^{(1)}\otimes...\otimes\theta^{(n)}\cdot\tau)=\tr(\WW_\tau\varrho)$. Here the equality is obtained by plugging in the definition~\eqref{eq:WitGenkn} of $\WW_\tau$, and we assume $\varrho:=\varrho^{(1)}\otimes...\otimes\varrho^{(n)}$ where each $\varrho^{(j)}$ is a product $(k_j-1)$-partite state. By convexity, $\tr(\WW_\tau\varrho)\geq 0$ also holds for any separable $(k_1+...+k_n-n)$-partite state $\varrho$.

\end{proof}
By the Choi-Jamio{\l}kowski isomorphism, each state or witness opertor $R^{(j)}$ can be seen as the Choi operator of a map $\Psi_j(\varrho^T):=\tr_{1...k_j-1}(R^{(j)}\cdot\varrho_{1...k_j-1}\otimes\one_{k_j})$ from systems $1...{k_j-1}$ to the system ${k_j}$, whose output is positive semidefinite if the inputs are products of positive semidefinite operators. 
Further discussion focused on the case of trace polynomials is given in~\cite{long}.

It is worth mentioning that this technique is not limited to using full-separability witnesses, in which case $\WW_\tau$ is a full-separability witness as well. As an example beyond, we shall consider a bipartite state or witness $\tau_{A_4B_4}$ and two entanglement witnesses $W_{A_1A_2A_3A_4}$ and $V_{B_1B_2B_3B_4}$ detecting non-$2|1|1$-separability (i.e. have nonnegative expectation value on states where two parties are entangled and the rest are separable). In particular, let $W$ and $V$ have a nonnegative expectation value on four partite states of the form $\varrho_{1234}=\varrho_{123}\otimes\varrho_4$ and convex combinations thereof, where $\varrho_{123}=p_1\varrho_1\otimes\varrho_{23}+p_2\varrho_2\otimes\varrho_{13}+p_3\varrho_{12}\otimes\varrho_3$ with $p_1,\,p_2,\,p_3\geq 0$ and $p_1+p_2+p_3=1$ (namely $\varrho_{123}$ is biseparable). This gives rise to the $(\kappa=6)$-partite entanglement witness
\begin{align}
    \WW_\tau = \tr_{A_4B_4}\Big (&W_{A_1A_2A_3A_4}\otimes V_{B_1B_2B_3B_4}\cdot\\
    &\quad\quad\quad\one_{A_1A_2A_3B_1B_2B_3}\otimes\tau_{A_4B_4}\Big )\nonumber\,.
\end{align}
This witness has a nonnegative expectation value on states of the form $\varrho_{A_1A_2A_3B_1B_2B_3}=\sum_iq_i\varrho^{(i)}_{A_1A_2A_3}\otimes\varrho^{(i)}_{B_1B_2B_3}$ where $q_i\geq 0$, $\sum_iq_i=1$, and $\varrho^{(i)}_{A_1A_2A_3}$ and $\varrho^{(i)}_{B_1B_2B_3}$ are biseparable, and detects states out of this set if $\tau$ is chosen such that $\WW_\tau$ has a negative eigenvalue. 
\\

\begin{figure}[]
    \centering
    \includegraphics[width=0.5\linewidth]{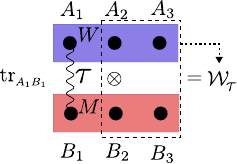}
    \caption{{\bf State-witness contraction yielding a four-partite witness} from three- and bi- partite states and witnesses, with three-partite systems $S^1:=A$ and $S^2:=B$. First we take the tensor product of a witness $W$ and a positive operator $M$. Then the first (or last, equivalently) subsystems $A_1$ and $B_1$ are contracted (multiplied and partial traced) with a state or witness $\tau_{A_1B_1}$ to be chosen -- see Eqs.~\eqref{eq:WitGenkn} and~\eqref{eq:Wit4partNonDeco}. This allows to lift $W$ to a larger witness $\WW_\tau$ with further detection capabilities (Observation~\ref{obs:4PartWitDetBoundEnt} and Proposition~\ref{prop:MPPTwits}), which moreover can be partially optimized over $\tau$ to detect specific four-partite states -- see Eq.~\eqref{eq:OptimWitGen}.}
    \label{fig:ContrWit}
\end{figure}

{\em Locally bound entanglement.} \,\,
As a further case of interest, we will employ this method to reuse existing highly symmetric witnesses given in~\cite{MaassenSlides}, which belong to a larger class of methods involving permutations in the Hilbert space~\cite{Elben2020,Neven_SymmetryResMomentsPT2021}. While these are desired due to their symmetries and implementation, their use in detecting states with local positive partial transpositions is a current challenge~\cite{Huber2021MatrixFO,Huber2022DimFree,MaassenSlides,PmapsWBAlg_Miqueleta2024,ImaiBoundfromRM_2021,liu2022MomentsPermutations,zhang2023experimentalverificationboundmultiparticle,CollectiveRM_Imai2024}. In particular, consider the witness
\begin{equation}\label{eq:Wit2deco}
 W = \frac{P_{2,1}}{4}-P_{1^3}\,
\end{equation}
constructed from the so-called Young projectors 
$P_{2,1}=  \eta_d\big (2(\id)-(123)-(132)\big )/6$ and $P_{1^3}= \eta_d\big ((\id)-(12)-(13)-(23)+(123)+(132)\big )/6$, where $\eta_d(\pi)\ket{v_1}\ket{v_2}\ket{v_3}=\ket{v_{\pi^{-1}(1)}}\ket{v_{\pi^{-1}(2)}}\ket{v_{\pi^{-1}(3)}}\in{\C^d}^{\otimes 3}$ represents the permutation $\pi\in S_3$ onto ${(\C^d)^{\otimes 3}}$ and $(\id)$ is the identity permutation. Here permutations are written in cyclic notation and Young projectors $P_\lambda$ are labeled by partitions $\lambda\vdash n$ (see Appendix~\ref{App:YoungProjectors} for further details). This witness can be derived from immanant inequalities~\cite{marshall1979inequalities,Lieb_Permanent_Conj_1996} using the techniques in~\cite{MaassenSlides,Huber2021MatrixFO,TP_Rico24}. Linear optimization shows that it can detect three-partite states with zero, one or two positive partial transpositions, but not three (see Table II in~\cite{long}). Using Proposition~\ref{prop:StWitContr} with an ancillary Bell state $\ket{\phi^+}=\sum_{i=0}^{2}\ket{ii}/\sqrt{3}$, it can be recycled to detect four-partite states whose local parties have positive partial transpositions, and thus cannot be distilled with local operations with classical communication if the four parties are spatially separated. This is a type of unlockable bound entanglement~\cite{4partUnlockBound_Smolin2001,SepDistMultiPartBoundUnlock_Dur1999,GenSmolinStatesUnlockBound_Augusiak2006,BoundMaxViolBellUnlockApplic_Augusiak2006,QComComplexBoundUnlockApplic_Bruckner2002,RemoteInfoConcUnlockBoundApplic_Murao2001}, which we refer to as {\em locally bound}.  \begin{observation}\label{obs:4PartWitDetBoundEnt}
Let $W$ be the witness of Eq.~\eqref{eq:Wit2deco}, and let $P:=P_{1^3}$ be the three-qutrit antisymmetrizer. Choose  $\tau=\dyad{\phi^+}$. The four-qutrit entanglement witness
 \begin{equation}\label{eq:Wit4partNonDeco}
  \WW_{\tau} = \tr_{A_3B_3}\Big (W_{A_1A_2A_3}\otimes P_{B_1B_2B_3}\cdot\one\otimes\tau_{A_3B_3}\Big )
 \end{equation}
 can detect non-fully-separable four-partite states $\varrho_{A_1A_2B_1B_2}$ with $\varrho^{T_{i}}\geq 0$ for $i\in\{A_1,A_2,B_1,B_2\}$. In particular, $\min\{\tr(\WW_{\tau}):\varrho^{T_{i}}\geq 0\}\approx-0.000276$\,.
\end{observation}

Note that this can be achieved using the shared quantum resources provided by the Bell state in an ancillary system $A_3B_3$. Yet, these can be reduced significantly: using a less entangled state $\ket{\phi_s}= \sqrt{s}\ket{00}+\sqrt{1-s}\ket{11}$ embedded into $\C^3\otimes\C^3$ with $s=0.9999$, a minimum value $\min_{\varrho^{T_i}\geq 0}\langle\WW_\tau\rangle_\varrho\approx -6.051\cdot 10^{-5}$ is obtained using only $-s\log_2(s) - (1-s)\log_2(1-s)\approx 0.0015$ ebits of entanglement.
\\

{\em Multipartite bound entanglement.}\,\,
Here we consider the detection of multipartite entangled states that cannot be distilled to pure entanglement across any of the bipartitions. This is the case of multipartite entangled states whose partial transpositions across all bipartitions are positive semidefinite (MPPT). Our goal is to lift non-decomposable multipartite witnesses into larger ones, also non-decomposable. Namely, to prove the following:
\begin{proposition}\label{prop:MPPTwits}
 Using SWC, it is possible to use an entanglement witness detecting $k$-partite MPPT states of certain size, to create a new witness detecting $k'>k$-partite MPPT states.
\end{proposition}
The proof shall be shown with  an analytical example. Consider a $k$-partite witness written in the computational basis, $W=\sum_{i_1,j_i,...,i_k,j_k=0}^{d-1}w_{i_1j_1...i_kj_k}\ket{i_1}\bra{j_1}\otimes...\otimes\ket{i_k}\bra{j_k}$; and the unnormalized $l$-partite GHZ state $\ket{\GHZ_{l,d}}=\sum_{t=0}^{d-1}\ket{t}^{\otimes l}$. We shall use SWC with $M=\dyad{\GHZ_{l,d}}$ to obtain a $\kappa$-partite witness $\WW=\tr_{k,k+l}\big (W_{1...k}\otimes M_{k+1...k+l}\cdot \one\otimes \tau_{k,k+l}\big )$, where $\kappa=k+l-2$. By employing $\tau$ to be an unnormalized maximally entangled state $d\dyad{\phi^+}$ or the SWAP operator $\eta_d(12)=d\dyad{\phi^+}^{T_2}$, one obtains respectively the witness
\begin{equation}\label{eq:MPPTwit}
 \WW_{\phi^+}=\sum_{\vec{i},\vec{j}\in\mathbf{Z}_d^{k}}w_{\vec{i},\vec{j}}\ket{i_1}\bra{j_1}\otimes...\otimes\ket{i_{k}}\bra{j_{k}}^{\otimes l-1}
\end{equation}
with $\vec{i}=i_1,...,i_{k}$ and $\vec{j}=j_1,...,j_{k}$,
or its partial transposition $\WW_{(12)}=\WW_{\phi^+}^{T_{k,...,k+l-2}}$. 
If the $k$-partite witness used is the three-qubit witness $W_I$ dectecting locally PPT states given in~\cite{kye2015threeIndecWit}, then we can find a new entanglement witness $\WW_{\phi^+}$ detecting MPPT states shared among four and five qubits. This is detailed in Appendix~\ref{app:LiftIndecWit}, together with the obtention of less sparse witnesses for MPPT states through different choices of the operator $M$.
\\

{\em Witnesses tailored to states.}\,\, 
While in Proposition~\ref{prop:StWitContr} one can choose $\tau$ to be any block-positive operator (state or witness), the optimization becomes accessible with semidefinite programming if $\tau$ is restricted to be a state or a decomposable witness~\cite{anjos2011handbookConicOpt} (or indecomposable, within the first levels of symmetric extension~\cite{doherty2004complete,navascues2008convergent}). Consider a fixed set of $n$ $k_j$-partite witnesses (or states) $R^{(i)}$ and a target $(\kappa=k_1+...+k_n-n)$-partite state $\varrho$ in local dimension $d$ to be detected. Let us choose $\tau$ to be optimized over all possible states and decomposable witnesses~\cite{EntDetRev_Guhne2009},
\begin{align}
    \tau^*=\text{argmin}_\tau\quad &\tr\bigg (\bigotimes_{i=1}^{n}R^{(i)}\cdot \varrho\otimes\tau\bigg )\label{eq:OptimWitGen}\\
    \text{so that}\quad &\tau = X + \sum_{\mathcal{S}\subset\{k_1,...,k_n\}}X_{\mathcal{S}}\nonumber\\
    &X,X_{\mathcal{S}}^{T_{\mathcal{S}}}\geq 0\nonumber\\
    &\tr(\tau)=1.\nonumber
\end{align}
Here $\mathcal{S}$ denotes any possible subset of the last subsystems $\{k_1,...,k_n\}$ and labels the operators $X_{\mathcal{S}}$ acting on $\C^{d_{k_1}}\otimes\cdots\otimes\C^{d_{k_n}}$, and thus the second line imposes that $\tau$ is a decomposable $n$-partite entanglement witness (or a state, if optimality is achieved for $X_{\mathcal S}=0$). Notice that the variable $\tau$ is of size $d^n$~\cite{parrilo2003SDP}, even though the target state to be detected is described by a density matrix of size $d^{(k-1)n}$ (assuming $k_1=...=k_n:=k$ for simplicity).

As a simple example, using that $\tr(\phi^+\cdot X\otimes Y)=\tr(X^T Y)$ and $\tr\big (\eta_d(12)\cdot X\otimes Y\big )=\tr(XY)$ where $\phi^+=\dyad{\phi^+}$ and $\eta_d(12)$ is the SWAP operator, a witness given by
\begin{equation}\label{eq:WitnesSDPforNPT}
    \WW_{\tau^*} = \tr_{A_2B_2}\Big ( \eta_d(12)_{A_1A_2}\otimes\phi^+_{B_1B_2}\cdot\varrho_{A_1B_1}\otimes\tau^*_{A_2B_2} \Big )
\end{equation}
with $\tau^*=\text{argmin}_\tau\big \{\tr (\WW_\tau\varrho ):\tau\geq 0,\tr\tau=1\big \}$ detects any state $\varrho$ with nonpositive partial transpose. Similarly, standard witness tailoring techniques using diagonalization~\cite{EntDetRev_Guhne2009} are recovered with this example: the duality between positive maps and witnesses~\cite{GeoQuantStates2006} is used to detect any state detected by a positive map $\Lambda_A\otimes\id_B(\varrho_{AB})\not\geq 0$, by optimizing over $\tau^*$ and replacing $\phi^+$ in Eq.~\eqref{eq:WitnesSDPforNPT} for $\Lambda^\dag\otimes\id(\eta_d(12))$ where $\Lambda^\dag$ satisfies $\tr(\Lambda(X)Y)=\tr(X\Lambda^\dag(Y))$ for all operators $X$ and $Y$.
\\

{\em Multiple copies.\,\,}
Detecting entanglement in copies of a state can be very advantageous with respect to linear methods~\cite{RemikUniversal2008,HoroMulticopyWit2003,Liu_FundamLimDet2022}. However, nonlinearity makes it hard to tailor such techniques to desired states. Here we show how 
Eq.~\eqref{eq:OptimWitGen} allows to partially solve this problem with simple semidefinite program whose size is only the size of the state. 
As a case study, we shall 
significantly improve our previous trace-polynomial multicopy constructions given in~\cite{TP_Rico24}: notice that the latter are obtained 
by setting in Proposition~\ref{prop:StWitContr} the input state to be composed of $k-1$ copies of a $\kappa$-partite target state $\varrho^{\otimes k-1}$ with $\kappa=2(k-1)$ and the naive choice $\tau=\varrho$. For example, consider a Bell state affected by white noise,
\begin{equation}\label{eq:Isotropic}
   \varrho = p\dyad{\phi^+} + (1-p)\frac{\one}{d^2}\,,
\end{equation}
which is often used as a worse-case scenario in noisy setups~\cite{HoroRedCrit_1999}. 
In~\cite{TP_Rico24} we showed that this state can be detected for a range of values of $p$ with the nonlinear witness $\WW=(\one-6P_{1^3})_{A_{123}}\otimes (P_{1^3})_{B_{123}}$ acting on three copies of $\varrho$, through $0<\tr(\WW\varrho^{\otimes 3})=\tr(\WW_\varrho \varrho^{\otimes 2})$ where we denote $\WW_\varrho:=\tr_{A_3B_3}(\WW\cdot\one\otimes\varrho_{A_3B_3})$ for easier comparison to our present work. From Proposition~\ref{prop:StWitContr} it becomes clear that $\WW$ can be better tuned to $\varrho$ using the freedom $\tau\neq\varrho$ given by Eq.~\eqref{eq:OptimWitGen}, in this case leading to (Fig.~\ref{fig:ContrWit})
\begin{equation}\label{eq:WitP(1-P)exampleOpt}
    \WW_\tau := \tr_{A_3B_3}\Big (\WW\cdot(\one_{A_1A_2B_1B_2}\otimes\phi^+_{A_3B_3})\Big )\,.
\end{equation}
In Table~\ref{tab:DetectIso} we compare the smallest values of $p$ so that states of Eq.~\eqref{eq:Isotropic} are detected with $\WW_\varrho$ (obtained in~\cite{TP_Rico24}) and the optimized witness $\WW_{\phi^+}$ in~\eqref{eq:WitP(1-P)exampleOpt}.

\begin{figure}
    \centering
    \includegraphics[width=0.7\linewidth]{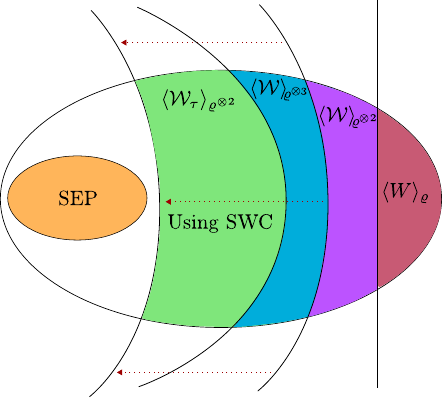}
    \caption{{\bf Tailoring multicopy witnesses to states.} Linear witnesses detect regions of entangled states separated by hyperplanes (magenta area) through $\langle W\rangle_{\varrho}=\tr(W\varrho)<0$. Multicopy witnesses $\WW$ detect regions of entangled states through $\langle \WW\rangle_{\varrho^{\otimes k}}<0$~\cite{HoroMulticopyWit2003,RemikUniversal2008}, separated by curved hypersurfaces (purple and blue areas). In particular, $\WW$ can be constructed by tensoring linear witnesses~\cite{TP_Rico24}. By replacing one of the copies by a variable state or witness $\tau$, the state-witness contraction technique (SWC) allows to systematically detect a larger set of states (green area --see Table~\ref{tab:DetectIso} for an example) and approach further the set of separable states (SEP).}
    \label{fig:NonlinOpt}
\end{figure}

Note that this method of tailoring witnesses of the form proposed in~\cite{TP_Rico24} preserves their invariant properties:
\begin{observation}
    Let $\WW_\tau$ in Eq.~\eqref{eq:OptimWitGen} be constructed from trace polynomial operators $R^{(j)}$ obtained by immanant inequalities according to~\cite{MaassenSlides}. Then, the minimization of $\tr(\WW_\tau\varrho^{\otimes k-1})$ over $\tau$ is invariant under local unitaries acting on $\varrho$.
\end{observation}
\begin{proof}
Local unitaries act as $U^{\otimes k-1}$ onto the first $k
-1$ copies of each local party. A minimum obtained for $\varrho_{AB}$ by  $\tau_{AB}$ is also obtained for $U\otimes V\varrho_{AB}U^\dag\otimes V^\dag$ by $U\otimes V\tau_{AB}U^\dag\otimes V^\dag$ (extension to the multipartite case is straightforward). Thus local unitaries act in practice as $U^{\otimes k}$ on each party. The rest of the argument is the same as in~\cite{TP_Rico24}, namely $U^{\otimes k}$ commutes with the representation $\eta_d(\pi)$ of permutations $\pi\in S_{k}$ due to the Schur-Weyl duality.
\end{proof}

\begin{table}[]
    \centering
    \begin{tabular}{c c c c c}
    \hline
      $p^*$:\, detect \quad & \quad $d=3$ & \quad $d=4$ & \quad $d=5$ & \quad $d=6$ \\
    \hline
       $\langle\WW\rangle_{\varrho^{\otimes^{3}}}<0$~\cite{TP_Rico24} \quad & \quad $0.727$ \quad & \quad $0.674$ & \quad $0.630$ & \quad $0.594$ \\
       $\langle\WW_{\phi^+}\rangle_{\varrho^{\otimes^{2}}}<0$~\eqref{eq:OptimWitGen}\,\,\, \quad & \quad $0.319$ \quad & \quad $0.262$ & \quad $0.222$ &\quad $0.193$ \\
       $\varrho\in\text{ENT}$ \quad & \,\, $0.25$ & \quad $0.2$ & \quad $0.1\overline{6}$ & \quad $0.\overline{142857}
       $ \\
    \hline
    \end{tabular}
    \caption{{\bf Advantage of tailored nonlinear witnesses}, over non-tailored. In each column we display the threshold parameter $p^*$ so that the sate~\eqref{eq:Isotropic} is detected for all values $p\geq p^*$, obtained numerically, with different methods for local dimension $d=3,4,5,6$. We compare these values corresponding to the non-tailored witness $\WW$ we provided in~\cite{TP_Rico24} (first row, $\tau=\varrho$); and the tailored version of this witness~\eqref{eq:WitP(1-P)exampleOpt} (second row, $\tau=\dyad{\phi^+}$). In this toy model, the state is entangled (ENT) for $p>(d+1)^{-1}$ and separable otherwise~\cite{HoroRedCrit_1999}.}
    \label{tab:DetectIso}
\end{table}
Besides being unitary invariant, the witnesses we constructed in~\cite{TP_Rico24} can be implemented with randomized measurements~\cite{Elben2020,Elben22Toolbox,Neven_SymmetryResMomentsPT2021,Huang_ManyPropsFewMeas_2020}. They key idea is that since only moments of the density matrix $\varrho$ are involved in the expectation values, then postprocessing the classical shadows can be done by matrix multiplication, thus avoiding large arrays. Note that the improvement in Eq.~\eqref{eq:OptimWitGen} preserves this implementability: one replaces each $k$'th classical shadow $\varrho^{r_k}=\bigotimes_{i=1}^n\varrho^{k}_i$ obtained from a set of random outcomes $r_k$ by the known optimal state $\tau$, and operates analogously to approximate $\tr(\WW_\tau\varrho^{\otimes k-1})$.
\\

{\em Distillability.}\,\,
Besides detecting entanglement, it is a major problem to detect whether a given bipartite state can be distilled to a two-qubit Bell pair by using multiple copies~\cite{HoroRedCrit_1999,DistBip_Dur2000}. Here we address this problem as follows: in Proposition~\ref{prop:StWitContr} consider the case $S^1=A$ and $S^2=B$, both with $k$ parties having the same dimensions $d_A$ and $d_B$, similarly as in the multi-copy scenario above. Let us set $R^{(1)}=P_A\Pi^{\otimes k}$ and $R^{(2)}=Q_B\Pi^{\otimes k}$, where $P_A$ and $Q_B$ are projectors onto subspaces of ${(\C^{d_{A}})}^{\otimes k}$ and ${(\C^{d_{B}})}^{\otimes k}$, and $\Pi$ is a local projector onto $\C^2$ (see Tab. I and Fig. 2 in~\cite{long} for specific results derived in this case). Then the $(k-1)$-copy action of the multilinear positive maps $\Psi_{A_1...A_{k-1}}$ and $\Psi_{B_1...B_{k-1}}$ defined by the multipartite Choi matrices $R^{(1)}$ and $R^{(2)}$, namely $\Psi_{A_1...A_{k-1}}\otimes\Psi_{B_1...B_{k-1}}(\varrho_{AB}^{\otimes k-1})$, projects $k-1$ copies of a bipartite state to the two-qubit subspace $\C^2\otimes\C^2$, where the PPT-criterion detects both entanglement and distillability~\cite{HorodeckiPPT_1996,HoroRedCrit_1999}. 

Thus the framework introduced in this work naturally provides a criterion for $(k-1)$-copy distillability, since two-qubit entanglement can be distilled from a state $\varrho_{AB}$ if and only if it can be transformed into a two-qubit entangled state through local projections on $(k-1)$ copies~\cite{Evidence_Vicenzo2000,HorodeckiBoundEnt_1998}. To assure that this is the case, it is sufficient to find local projections on $(k-1)$ copies whose output is a state with non-positive partial transposition. In particular, this is guaranteed if 
\begin{equation}\label{eq:Distillability}
    \min_\tau\tr\Big ((P_A\cdot\Pi_{A_1}^{\otimes k})\otimes (P_B\cdot\Pi_{B_1}^{\otimes k})\cdot \varrho_{A_1B_1}^{\otimes k-1}\otimes\tau_{A_kB_k}\Big )<0.
\end{equation}
Here the minimization is taken over two-qubit unit trace Hermitian operators $\tau$ with positive partial transpose $\tau^{T_A}\geq 0$, to optimize over all two-qubit witnesses. If strict negativity in Eq.~\eqref{eq:Distillability} is satisfied, then the projected output has non-positive partial transposition and therefore $\varrho_{AB}$ is $(k-1)$-copy distillable. Here the notation $P_A$ and $P_B$ highlights that these operators have support across the $k$ copies of $A_1$ and $B_1$.

This shall be illustrated in specific examples involving a single copy, considering the two-qudit antisymmetrizer $P_{1,1}^{(d)}=(\eta_d(\id)-\eta_d(12))/2$ and a two-dimensional local projector $\Pi=\diag(1,1,0,...,0)$ diagonal in the computational basis: in~\cite{long} we demonstrate 
that most generic states in a random sample are detected to be one-copy distillable by using $R^{(1)}=P_{1,1}\Pi^{\otimes 2}$ (see~\cite{long} for a detailed numerical analysis), effectively detecting the qubit-qudit distillable entanglement present in a subspace of a two-qudit system. Moreover, choosing $R^{(1)}=R^{(2)}=P_{1,1}\Pi^{\otimes 2}$, the operator $\WW = R^{(1)}_{A_1A_2}\otimes R^{(2)}_{B_1B_2}$ detects both entanglement and one-copy distillability on all Werner states $\varrho_W$ known to be distillable, namely
\begin{equation}
    \varrho_W = p P_{2,0}^{(d)}/d_s + (1-p)P_{1,1}^{(d)}/d_a 
\end{equation}
within $p\in(p_0,1/2]$ with $p_0=\frac{d+1}{4d-2}$~\cite{DistBip_Dur2000}, where $P_{2,0}^{(d)}=(\eta_d(\id)+\eta_d(12))/2$,  $d_s=d(d+1)/2$ and $d_a=d(d-1)/2$. Indeed, one verifies that
\begin{equation}\label{ProjectToWerner}
    \tr_{A_1B_1}\Big ( \WW \cdot {\varrho_W}_{A_1B_1}\otimes\one_{A_2B_2} \Big ) = p P_{2,0}^{(2)}/d_s + (1-p) P_{1,1}^{(2)}/d_a\,,
\end{equation}
where $P_{2,0}^{(2)}$ and $P_{1,1}^{(2)}$ are the two-qubit symmetrizer and antisymmetrizer (projectors onto the triplet and singlet subspaces) respectively. Therefore, choosing the two-qubit SWAP operator $\tau=2\dyad{\phi^+}^{T_A}=\eta_2(12)$, we 
have $\tr(\WW_\tau\varrho_W)<0$ for $p\in(p_0,1/2]$. This successful ansatz shows that the presented method to detect entanglement is also well suited to detect distillability, by using local projectors and optimization over two-qubit witnesses. This approach recovers a similar spirit as earlier criteria using witnesses~\cite{Kraus02DistEW} or SDP optimization~\cite{Rains06DistWernerSDP,Doherty06DistWernerSDP,Wang16ImprovedDistWernerSDP,Rozp18DistOpt}. 
\\

{\em Conclusions.}\,\, 
We have developed a technique to design new linear and nonlinear entanglement witnesses, and demonstrated its capabilities with examples. As a proof of concept we have shown that 
limited shared resources allow to employ witnesses unable to detect locally bound entangled states
to construct new larger witnesses able to do so; and that available examples of witnesses for few-body bound entangled states can be used to construct new witnesses detecting bound entanglement in larger systems.  
The method applies to construct nonlinear witnesses acting on multiple copies.
We are able to optimize these to target states with little computational resources, which is a current challenge. As a case study of both theoretical and experimental interest, we have seen that our previous multicopy construction provided in~\cite{TP_Rico24} can be tuned to significantly better detect noisy entangled states, while the desired properties of unitary invariance and implementability with randomized measurements have been shown to be preserved. Moreover, the technique introduced leads to practical recipes to detect distillability, which are demonstrated to be effective for relevant families of states.

\bigskip

{\em Acknowledgements:\,\,} The author is thankful to 
Pawe{\l} Horodecki, 
Felix Huber, 
Robin Krebs, 
Ferran Riera-Sàbat, 
Some Sankar Bhattacharya and 
Anna Sanpera 
for comments and support. An anonymous Referee is acknowledged for asking whether entanglement structures beyond full-separability can be detected. Support by the Foundation for Polish Science through TEAM-NET project  POIR.04.04.00-00-17C1/18-00
and by NCN QuantERA
Project No. 2021/03/Y/ST2/00193 is gratefully acknowledged. The author also acknowledges financial support from Spanish MICIN (projects: PID2022:141283NBI00;139099NBI00) with the support of FEDER funds, the Spanish Goverment with funding
from European Union NextGenerationEU (PRTR-C17.I1), the Generalitat de Catalunya,
the Ministry for Digital Transformation and of Civil Service of the Spanish Government through the QUANTUM ENIA project -Quantum Spain Project- through the Recovery, Transformation and Resilience Plan NextGeneration EU within the framework
of the Digital Spain 2026 Agenda.   

\addcontentsline{toc}{subsection}{Bibliography}
\bibliographystyle{ieeetr}
\bibliography{Bibliography}{}

\begin{thebibliography}{10}

\bibitem{advancesHDent_Erhard2020}
M.~Erhard, M.~Krenn, and A.~Zeilinger, ``Advances in high-dimensional quantum
  entanglement,'' {\em Nature Reviews Physics}, vol.~2, no.~7, pp.~365--381,
  2020.

\bibitem{EfficientLargeScaleMBdyn_Artaco2024}
C.~Artiaco, C.~Fleckenstein, D.~Aceituno~Ch\'avez, T.~K. Kvorning, and J.~H.
  Bardarson, ``Efficient large-scale many-body quantum dynamics via
  local-information time evolution,'' {\em PRX Quantum}, vol.~5, p.~020352, Jun
  2024.

\bibitem{MPEntSuperQubits_Lu2022}
M.~Lu, J.-L. Ville, J.~Cohen, A.~Petrescu, S.~Schreppler, L.~Chen, C.~J\"unger,
  C.~Pelletti, A.~Marchenkov, A.~Banerjee, W.~P. Livingston, J.~M. Kreikebaum,
  D.~I. Santiago, A.~Blais, and I.~Siddiqi, ``Multipartite entanglement in
  rabi-driven superconducting qubits,'' {\em PRX Quantum}, vol.~3, p.~040322,
  Nov 2022.

\bibitem{LargScaleQNetworksKozlowski_2019}
W.~Kozlowski and S.~Wehner, ``Towards large-scale quantum networks,'' in {\em
  Proceedings of the Sixth Annual ACM International Conference on Nanoscale
  Computing and Communication}, NANOCOM ’19, ACM, Sept. 2019.

\bibitem{NetworkGMPE_Navascues2020}
M.~Navascu\'es, E.~Wolfe, D.~Rosset, and A.~Pozas-Kerstjens, ``Genuine network
  multipartite entanglement,'' {\em Phys. Rev. Lett.}, vol.~125, p.~240505, Dec
  2020.

\bibitem{DetMPentManyBody_Frerot2022}
I.~Fr\'erot, F.~Baccari, and A.~Ac\'{\i}n, ``Unveiling quantum entanglement in
  many-body systems from partial information,'' {\em PRX Quantum}, vol.~3,
  p.~010342, Mar 2022.

\bibitem{GURVITS2003sepNPhard}
L.~Gurvits, ``Classical complexity and quantum entanglement,'' {\em Journal of
  Computer and System Sciences}, vol.~69, no.~3, pp.~448--484, 2004.
\newblock Special Issue on STOC 2003.

\bibitem{DetMBodyContinuous_Kunkel2022}
P.~Kunkel, M.~Pr\"ufer, S.~Lannig, R.~Strohmaier, M.~G\"arttner, H.~Strobel,
  and M.~K. Oberthaler, ``Detecting entanglement structure in continuous
  many-body quantum systems,'' {\em Phys. Rev. Lett.}, vol.~128, p.~020402, Jan
  2022.

\bibitem{EntDetRev_Guhne2009}
O.~Gühne and G.~Tóth, ``Entanglement detection,'' {\em Physics Reports},
  vol.~474, p.~1–75, Apr. 2009.

\bibitem{ReviewQEnt_Horo2009}
R.~Horodecki, P.~Horodecki, M.~Horodecki, and K.~Horodecki, ``Quantum
  entanglement,'' {\em Rev. Mod. Phys.}, vol.~81, pp.~865--942, Jun 2009.

\bibitem{bae2020mirrored}
J.~Bae, D.~Chru{\'s}ci{\'n}ski, and B.~C. Hiesmayr, ``Mirrored entanglement
  witnesses,'' {\em npj Quantum Information}, vol.~6, no.~1, p.~15, 2020.

\bibitem{doherty2004complete}
A.~C. Doherty, P.~A. Parrilo, and F.~M. Spedalieri, ``Complete family of
  separability criteria,'' {\em Phys. Rev. A}, vol.~69, no.~2, p.~022308, 2004.

\bibitem{Navascues_PPT_DPS_2009}
M.~Navascu\'es, M.~Owari, and M.~B. Plenio, ``Power of symmetric extensions for
  entanglement detection,'' {\em Phys. Rev. A}, vol.~80, p.~052306, Nov 2009.

\bibitem{Huber2022DimFree}
F.~Huber, I.~Klep, V.~Magron, and J.~Vol{\v{c}}i{\v{c}}, ``Dimension-free
  entanglement detection in multipartite {W}erner states,'' {\em Communications
  in Mathematical Physics}, vol.~396, pp.~1051--1070, aug 2022.

\bibitem{Dagmar_DetectHypergraph2017}
M.~Ghio, D.~Malpetti, M.~Rossi, D.~Bru{\ss}, and C.~Macchiavello,
  ``Multipartite entanglement detection for hypergraph states,'' {\em Journal
  of Physics A: Mathematical and Theoretical}, vol.~51, p.~045302, december
  2017.

\bibitem{Guhne_DetectGraph}
O.~G\"{u}hne, G.~T\'{o}th, P.~Hyllus, and H.~J. Briegel, ``Bell inequalities
  for graph states,'' {\em Physical Review Letters}, vol.~95, p.~120405, Sep
  2005.

\bibitem{Liu_FundamLimDet2022}
P.~Liu, Z.~Liu, S.~Chen, and X.~Ma, ``Fundamental limitation on the
  detectability of entanglement,'' {\em Physical Review Letters}, vol.~129,
  p.~230503, Nov 2022.

\bibitem{HoroMulticopyWit2003}
P.~Horodecki, ``From limits of quantum operations to multicopy entanglement
  witnesses and state-spectrum estimation,'' {\em Physical Review A}, vol.~68,
  p.~052101, Nov 2003.

\bibitem{RemikUniversal2008}
R.~Augusiak, M.~Demianowicz, and P.~Horodecki, ``Universal observable detecting
  all two-qubit entanglement and determinant-based separability tests,'' {\em
  Phys. Rev. A}, vol.~77, p.~030301, Mar 2008.

\bibitem{TP_Rico24}
A.~Rico and F.~Huber, ``Entanglement detection with trace polynomials,'' {\em
  Phys. Rev. Lett.}, vol.~132, p.~070202, Feb 2024.

\bibitem{Karol-GeoQstates2006}
I.~Bengtsson and K.~\.Zyczkowski, {\em Geometry of Quantum States: An
  Introduction to Quantum Entanglement}.
\newblock Cambridge University Press, 2006.

\bibitem{long}
A.~Rico, ``Mixed state entanglement from symmetric matrix inequalities,'' {\em
  arXiv:2502.18446}, 2025.

\bibitem{MaassenSlides}
H.~Maassen and B.~K\"{u}mmerer, ``Entanglement of symmetric {W}erner states,''
  {\em Workshop: Mathematics of Quantum Information Theory,
  \url{http://www.bjadres.nl/MathQuantWorkshop/Slides/SymmWernerHandout.pdf}},
  2019.

\bibitem{Elben2020}
A.~Elben, R.~Kueng, H.-Y.~R. Huang, R.~van Bijnen, C.~Kokail, M.~Dalmonte,
  P.~Calabrese, B.~Kraus, J.~Preskill, P.~Zoller, and B.~Vermersch,
  ``Mixed-state entanglement from local randomized measurements,'' {\em
  Physical Review Letters}, vol.~125, p.~200501, Nov 2020.

\bibitem{Neven_SymmetryResMomentsPT2021}
A.~Neven, J.~Carrasco, V.~Vitale, C.~Kokail, A.~Elben, M.~Dalmonte,
  P.~Calabrese, P.~Zoller, B.~Vermersch, R.~Kueng, and B.~Kraus,
  ``Symmetry-resolved entanglement detection using partial transpose moments,''
  {\em npj Quantum Information}, vol.~7, oct 2021.

\bibitem{Huber2021MatrixFO}
F.~Huber and H.~Maassen, ``Matrix forms of immanant inequalities,'' {\em
  arXiv:2103.04317}, 2021.

\bibitem{PmapsWBAlg_Miqueleta2024}
M.~Balanzó-Juandó, M.~Studziński, and F.~Huber, ``Positive maps from the
  walled brauer algebra,'' {\em Journal of Physics A: Mathematical and
  Theoretical}, vol.~57, p.~115202, mar 2024.

\bibitem{ImaiBoundfromRM_2021}
S.~Imai, N.~Wyderka, A.~Ketterer, and O.~G\"uhne, ``Bound entanglement from
  randomized measurements,'' {\em Phys. Rev. Lett.}, vol.~126, p.~150501, Apr
  2021.

\bibitem{liu2022MomentsPermutations}
Z.~Liu, Y.~Tang, H.~Dai, P.~Liu, S.~Chen, and X.~Ma, ``Detecting entanglement
  in quantum many-body systems via permutation moments,'' {\em Physical Review
  Letters}, vol.~129, no.~26, p.~260501, 2022.

\bibitem{zhang2023experimentalverificationboundmultiparticle}
C.~Zhang, Y.-Y. Zhao, N.~Wyderka, S.~Imai, A.~Ketterer, N.-N. Wang, K.~Xu,
  K.~Li, B.-H. Liu, Y.-F. Huang, C.-F. Li, G.-C. Guo, and O.~Gühne,
  ``Experimental verification of bound and multiparticle entanglement with the
  randomized measurement toolbox,'' 2023.

\bibitem{CollectiveRM_Imai2024}
S.~Imai, G.~T\'oth, and O.~G\"uhne, ``Collective randomized measurements in
  quantum information processing,'' {\em Phys. Rev. Lett.}, vol.~133,
  p.~060203, Aug 2024.

\bibitem{marshall1979inequalities}
A.~W. Marshall, I.~Olkin, and B.~C. Arnold, {\em Inequalities: {T}heory of
  {M}ajorization and its {A}pplications}.
\newblock Springer, 1979.

\bibitem{Lieb_Permanent_Conj_1996}
E.~H. LIEB, ``Proofs of some conjectures on permanents,'' {\em Journal of
  Mathematics and Mechanics}, vol.~16, no.~2, pp.~127--134, 1966.

\bibitem{4partUnlockBound_Smolin2001}
J.~A. Smolin, ``Four-party unlockable bound entangled state,'' {\em Phys. Rev.
  A}, vol.~63, p.~032306, Feb 2001.

\bibitem{SepDistMultiPartBoundUnlock_Dur1999}
W.~D\"ur, J.~I. Cirac, and R.~Tarrach, ``Separability and distillability of
  multiparticle quantum systems,'' {\em Phys. Rev. Lett.}, vol.~83,
  pp.~3562--3565, Oct 1999.

\bibitem{GenSmolinStatesUnlockBound_Augusiak2006}
R.~Augusiak and P.~Horodecki, ``Generalized smolin states and their
  properties,'' {\em Phys. Rev. A}, vol.~73, p.~012318, Jan 2006.

\bibitem{BoundMaxViolBellUnlockApplic_Augusiak2006}
R.~Augusiak and P.~Horodecki, ``Bound entanglement maximally violating bell
  inequalities: Quantum entanglement is not fully equivalent to cryptographic
  security,'' {\em Phys. Rev. A}, vol.~74, p.~010305, Jul 2006.

\bibitem{QComComplexBoundUnlockApplic_Bruckner2002}
i.~c.~v. Brukner, M.~\ifmmode~\dot{Z}\else \.{Z}\fi{}ukowski, and A.~Zeilinger,
  ``Quantum communication complexity protocol with two entangled qutrits,''
  {\em Phys. Rev. Lett.}, vol.~89, p.~197901, Oct 2002.

\bibitem{RemoteInfoConcUnlockBoundApplic_Murao2001}
M.~Murao and V.~Vedral, ``Remote information concentration using a bound
  entangled state,'' {\em Phys. Rev. Lett.}, vol.~86, pp.~352--355, Jan 2001.

\bibitem{kye2015threeIndecWit}
S.-H. Kye, ``Three-qubit entanglement witnesses with the full spanning
  properties,'' {\em Journal of Physics A: Mathematical and Theoretical},
  vol.~48, no.~23, p.~235303, 2015.

\bibitem{anjos2011handbookConicOpt}
M.~F. Anjos and J.~B. Lasserre, {\em Handbook on semidefinite, conic and
  polynomial optimization}, vol.~166.
\newblock Springer Science \& Business Media, 2011.

\bibitem{navascues2008convergent}
M.~Navascu{\'e}s, S.~Pironio, and A.~Ac{\'\i}n, ``A convergent hierarchy of
  semidefinite programs characterizing the set of quantum correlations,'' {\em
  New J. Phys.}, vol.~10, no.~7, p.~073013, 2008.

\bibitem{parrilo2003SDP}
P.~A. Parrilo, ``Semidefinite programming relaxations for semialgebraic
  problems,'' {\em Mathematical programming}, vol.~96, pp.~293--320, 2003.

\bibitem{GeoQuantStates2006}
I.~Bengtsson and K.~{\.Z}yczkowski, {\em Geometry of Quantum States: an
  Introduction to Quantum Entanglement}.
\newblock Cambridge University Press, 2006.

\bibitem{HoroRedCrit_1999}
M.~Horodecki and P.~Horodecki, ``Reduction criterion of separability and limits
  for a class of distillation protocols,'' {\em Phys. Rev. A}, vol.~59,
  pp.~4206--4216, Jun 1999.

\bibitem{Elben22Toolbox}
A.~Elben, S.~T. Flammia, H.-Y. Huang, R.~Kueng, J.~Preskill, B.~Vermersch, and
  P.~Zoller, ``The randomized measurement toolbox,'' {\em Nature Reviews
  Physics}, vol.~5, pp.~9--24, dec 2022.

\bibitem{Huang_ManyPropsFewMeas_2020}
H.-Y. Huang, R.~Kueng, and J.~Preskill, ``Predicting many properties of a
  quantum system from very few measurements,'' {\em Nature Physics}, vol.~16,
  pp.~1050--1057, jun 2020.

\bibitem{DistBip_Dur2000}
W.~D\"ur, J.~I. Cirac, M.~Lewenstein, and D.~Bru\ss{}, ``Distillability and
  partial transposition in bipartite systems,'' {\em Phys. Rev. A}, vol.~61,
  p.~062313, May 2000.

\bibitem{HorodeckiPPT_1996}
M.~Horodecki, P.~Horodecki, and R.~Horodecki, ``Separability of mixed states:
  necessary and sufficient conditions,'' {\em Phys. Lett. A}, vol.~223,
  p.~1–8, Nov. 1996.

\bibitem{Evidence_Vicenzo2000}
D.~P. DiVincenzo, P.~W. Shor, J.~A. Smolin, B.~M. Terhal, and A.~V. Thapliyal,
  ``Evidence for bound entangled states with negative partial transpose,'' {\em
  Phys. Rev. A}, vol.~61, p.~062312, May 2000.

\bibitem{HorodeckiBoundEnt_1998}
M.~Horodecki, P.~Horodecki, and R.~Horodecki, ``Mixed-state entanglement and
  distillation: Is there a “bound” entanglement in nature?,'' {\em Physical
  Review Letters}, vol.~80, p.~5239–5242, June 1998.

\bibitem{Kraus02DistEW}
B.~Kraus, M.~Lewenstein, and J.~I. Cirac, ``Characterization of distillable and
  activatable states using entanglement witnesses,'' {\em Phys. Rev. A},
  vol.~65, p.~042327, Apr 2002.

\bibitem{Rains06DistWernerSDP}
E.~M. Rains, ``A semidefinite program for distillable entanglement,'' {\em IEEE
  Trans. Inf. Theor.}, vol.~47, p.~2921–2933, Sept. 2006.

\bibitem{Doherty06DistWernerSDP}
R.~O. Vianna and A.~C. Doherty, ``Distillability of werner states using
  entanglement witnesses and robust semidefinite programs,'' {\em Phys. Rev.
  A}, vol.~74, p.~052306, Nov 2006.

\bibitem{Wang16ImprovedDistWernerSDP}
X.~Wang and R.~Duan, ``Improved semidefinite programming upper bound on
  distillable entanglement,'' {\em Phys. Rev. A}, vol.~94, p.~050301, Nov 2016.

\bibitem{Rozp18DistOpt}
F.~Rozpedek, T.~Schiet, L.~P. Thinh, D.~Elkouss, A.~C. Doherty, and S.~Wehner,
  ``Optimizing practical entanglement distillation,'' {\em Phys. Rev. A},
  vol.~97, p.~062333, Jun 2018.

\bibitem{Ent3qubits_Dur2000}
W.~D\"ur, G.~Vidal, and J.~I. Cirac, ``Three qubits can be entangled in two
  inequivalent ways,'' {\em Phys. Rev. A}, vol.~62, p.~062314, Nov 2000.

\end{thebibliography}

\appendix

\section{Examples of Young projectors}\label{App:YoungProjectors}
The three-qudit Hilbert space $\mathcal{H}=\C^d\otimes\C^d\otimes\C^d$ decomposes into the symmetric, standard, and antisymmetric irreducible subspaces. These are labeled by the partitions $(3,0,0)$, $(2,1,0)$ and $(1,1,1)$, and are defined by the following Young projectors,
\begin{align}
    P_{3}&=\frac{1}{6}\eta_d\Big (\big [(\id)+(12)+(13)+(23)+(123)+(132)\big ]\Big )\,, \\
    P_{2,1}&= \frac{1}{6} \eta_d\Big (\big [2(\id)-(123)-(132)\big ]\Big )\,, \\
    P_{1^3}&= \frac{1}{6}\eta_d\Big (\big [(\id)-(12)-(13)-(23)+(123)+(132)\big ]\Big )\,.
\end{align}
Here permutations are written in cyclic notation. For example,
\begin{equation}
    \eta_d(123)\ket{v_1}\ket{v_2}\ket{v_3} = \ket{v_3}\ket{v_1}\ket{v_2}\,.
\end{equation}

\section{New witnesses for MPPT states}\label{app:LiftIndecWit}
Here we provide a simple example of an indecomposable three-partite entanglement witness lifted to a new entanglement witnesses detecting MPPT states in larger systems. We start from the three-qubit indecomposable witness
\begin{equation}
    W_I =
    \begin{pmatrix}
       \cdot & \cdot & \cdot & \cdot & \cdot & \cdot & \cdot & 1 \\
       \cdot & \cdot & \cdot & \cdot & \cdot & \cdot & 1 & \cdot \\
       \cdot & \cdot & \cdot & \cdot & \cdot & -1 & \cdot & \cdot \\
       \cdot & \cdot & \cdot & a & 1 & \cdot & \cdot & \cdot \\
       \cdot & \cdot & \cdot & 1 & b & \cdot & \cdot & \cdot \\
       \cdot & \cdot & -1 & \cdot & \cdot & \cdot & \cdot & \cdot \\
       \cdot & 1 & \cdot & \cdot & \cdot & \cdot & \cdot & \cdot \\
       1 & \cdot & \cdot & \cdot & \cdot & \cdot & \cdot & \cdot 
    \end{pmatrix}
\end{equation}
with $.=0$ and $ab=8$ given in~\cite{kye2015threeIndecWit}, which is able to detect three-qubit states with local positive partial transpositions. Notice that for four-partite states this condition is not enough for MPPT, as it does not ensure positivity of the partial transposition in subsets with more than one party. The method introduced here allows to do so. Using Eq.~\eqref{eq:MPPTwit} with $\tau=\phi^+$ and $M=\dyad{GHZ_{2,l}}$, for $l=3$ and $l=4$ we obtain respectively a 4-qubit and 5-qubit witness whose minimial expectation value over the set of MPPT states satisfies $\tr(\WW\varrho)<0$. By setting $M$ to be a $W$-state~\cite{Ent3qubits_Dur2000} or a Haar random state, a negative minimal expectation value is also obtained over the set of MPPT states.

\end{document}